\documentclass[a4paper,10pt]{article}
\usepackage[utf8x]{inputenc}
\usepackage{amsmath,amsfonts,amssymb}

\newtheorem{theorem}{Theorem}
\newtheorem{claim}[theorem]{Claim}
\newenvironment{proof}[1][Proof]{\noindent\textbf{#1:} }{\ \rule{0.5em}{0.5em}}

\title{Partition distances}
\author{ Giovanni Rossi\footnote{\textsc{contact}: \texttt{giorossi@cs.unibo.it} and/or \texttt{roxyjean@gmail.com}}\\
{\small Department of Computer Science, University of Bologna}\\
{\small Mura Anteo Zamboni 7, Bologna 40126, Italy} }

\begin{document}

\maketitle

\begin{abstract}
Alternative novel measures of the distance between any two partitions of a $n$-set are proposed and compared, together 
with a main existing one, namely \textit{partition-distance} $D(\cdot,\cdot)$. The comparison achieves by checking their
restriction to modular elements of the partition lattice, as well as in terms of suitable classifiers. Two of the new
measures obtain through the \textit{size}, a function mapping every partition into the number of \textit{atoms} finer than
that partition. One of these size-based distances extends to geometric lattices the traditional Hamming distance between
subsets, when these latter are regarded as hypercube vertexes or binary $n$-vectors. After carefully framing the environment,
a main comparison finally results from the following bounding problem: for every value $k$, with $0<k<n$, of
partition-distance $D(\cdot ,\cdot)$, determine the minimum and maximum of the \textit{indicator-Hamming} distance
$\delta^{IH}(P,Q)$ proposed here over all pairs of partitions $P,Q$ such that $D(P,Q)=k$.
\end{abstract}


\section{Introduction}
Over the last decade, considerable interest has been attracted on measuring the distance between partitions (as well as between
and/or within collections of partitions). The issue arises, in general, when making similarity comparisons between clusterings
\cite{SoberLook06,D'yachkov+++06,Gusfield02,Kono++05,MarinaMeila08,CentralPartition04,Yu+++01}.

The problem of quantifying the distance between partitions of a finite set is here approached with a specific combinatorial
target, in that the proposed measure aims at keeping into account the coarsening, meet and join relations of the partition
lattice exactly in the same way as the traditional Hamming distance between subsets does with inclusion, intersection and
union. Put it differently, the objective is reproducing the symmetric difference between subsets when measuring distances
between partitions.

Despite the analysis adopts such a focused and somehow theoretical perspective, still the outcome is a variety of novel
partition distance measures, each possibly meeting an alternative application need. In particular, the measure that
factually translates the traditional Hamming distance between subsets in terms of partitions appears to evaluate differences
in a very accurate and granular manner.

Meet, join and order relations of the subset and partition lattices, as well as their distinctive features and what
renders modular an element in a lattice, are described in \cite{Aigner79,Stanley71,Stern99}. In particular,
modular elements of the partition lattice are extensively dealt with in the sequel. Also, partitions are mostly treated as collections of
atoms of the partition lattice, and these atoms are modular. More generally, the approach leads to work with linear dependence
\cite{Whitney35}, commonly arising in geometric lattices. In a way, the indicator-Hamming distance measure
proposed below fully exploits such a linear dependence for evaluating differences between partitions.

The next section details two simple ways of translating the Hamming distance between subsets in terms of partitions: one is
through the symmetric difference while the other is through the rank. In section 3 they are compared with
\textit{partition-distance} proposed in \cite{Gusfield02} by checking their behavior over pairs of modular partitions. In
section 4 these three measures are characterized in terms of suitable classifiers (applying to any complemented lattice).
Section 5 focuses on atoms of the partition lattice, populating the first level of the Hasse diagram. The remainder of the
paper looks at partitions precisely in terms of their representations as a join of atoms. Linear dependence means that the
generic partition has many such representations. The size of a partition is the number of atoms finer than that partition or,
equivalently, the cardinality of the largest representation of that partition as a join of atoms \cite{RoxyJean04}. It is
shown to be a strictly monotone and super-modular partition function. Section 6 provides and characterizes two novel partition
distance measures: one is \textit{size-based}, using the size just like the rank-based distance (from section 2) uses the rank, while the other
is named \textit{indicator-Hamming} and proposed as the faithful translation of the Hamming distance between subsets. In
fact, it measures the distance between any two partitions by counting the number of atoms finer than either one but not both.
Section 7 details the features displayed by this IH distance measure by bounding its maximum and minimum for every
value of partition-distance. Essentially, apart from providing the sought combinatorial congruence,
the former distance is very precise and granular at quantifying differences between partitions: its range is large (much
larger than all those of other distances appearing here), and this is very useful for measuring distances between partitions
from the mostly populated levels of the Hasse diagram, where more distinct types of differences between partitions actually exist.
Final remarks are contained in sections 8.


\section{Symmetric difference and rank}
For a finite set $N=\{1,\ldots ,n\}$ (or $[n]$), let $(2^N,\cap ,\cup)$ and $(\mathcal P^N,\wedge, \vee)$ denote the corresponding
subset and partition lattices, with inclusion $\supseteq$ and coarsening $\geqslant$ as order relations, respectively. Both are
atomic, and the fomer is distributive while the latter is geometric indecomposable \cite{Aigner79,Stern99}.

The distance between elements of a ordered set is to be measured in terms of the order relation. On the other hand, measures of
the difference between elements of a generic set are commonly referred to as \textit{Hamming distances} when elements
are firstly represented as arrays, and next the difference between any two of them simply reduces to counting the number of
entries where their two array representations differ. In discrete settings, measuring distances seems to naturally reduce to counting.

The Hamming distance $d(A,B)$ between any two subsets $A,B\in 2^N$ is
\begin{equation}
d(A,B)=|A\Delta B|=|A\backslash B|+|B\backslash A|=r(A\cup B)-r(A\cap B)\text ,
\end{equation}
$r:2^N\rightarrow\mathbb Z_+$ being the rank function: $r(A)=|A|$ for all $A\in2^N$. In words, $d(\cdot,\cdot)$ counts how many $i\in N$ are
included in either $A$ or else $B$, but not in both. Note that such elements $i\in N$ are the atoms $\{i\}\in2^N$ of the subset
lattice. This is a Hamming distance in that subsets $A\in 2^N$ are firstly represented as binary vectors through their characteristic
function $\chi_A:N\rightarrow\{0,1\}$ defined by $\chi_A(i)=1$ if $i\in A$ and $\chi_A(i)=0$ if $i\in N\backslash A$, and next the
distance between any two subsets $A,B\in 2^N$ is the number of entries where $\chi_A$ and $\chi_B$ differ. That is, the
cardinality of their symmetric difference $A\Delta B$. 

Any subset $A\in 2^N$ has a unique complement $A^c=N\backslash A$. For all non-empty subsets
$\emptyset\subset A\subseteq N$ and all partitions $P\in\mathcal P^N$, denote by $P^A$ the partition of $A$ induced  by $P$,
and let $\mathcal P^A$ be the sub-lattice of partitions of $A$. Partition-distance
$D:\mathcal P^N\times\mathcal P^N\rightarrow\{0,1,\ldots ,n-1\}$ given by \cite{Gusfield02} is
\begin{equation}
D(P,Q)=\min\{|A^c|:\emptyset\subset A\subseteq N,P^A=Q^A\}\text .
\end{equation}
That is, the minimum number of elements $i\in N$ that must be deleted in order for the two residual induced partitions to
coincide. Also, $D(P,Q)$ \textit{is the minimum number of elements that must be moved between} [or away from] \textit{blocks of
$P$ so that the resulting partition equals $Q$} (see \cite[p. 160]{Gusfield02}). Although there exist Hamming distances between
partitions in the literature \cite{SoberLook06,MarinaMeila08}, partition-distance $D(\cdot,\cdot)$ is not among them, because in
(2) there is no count of non-matched entries in some array representations of $P$ and $Q$.  On the other hand, there are two
immediate ways of paralleling (1) when switching from subsets to partitions. One is treating partitions as special collections of
subsets, while the other is using the rank of the partition lattice just like $r(\cdot)$ appears in (1). These two alternatives
are now briefly detailed.

Partitions may well be looked at as subsets of $2^N$, in that $P\subset 2^N$ or equivalently $P\in 2^{2^N}$ for all
$P\in\mathcal P^N$. Hence, the distance $\delta^{SD}(P,Q)$ between any two partitions $P$ and $Q$ may be measured as the
cardinality
\begin{equation}
\delta^{SD}(P,Q)=|P\Delta Q|=|P\backslash Q|+|Q\backslash P|=|P\cup Q|-|P\cap Q|
\end{equation}
of their \textit{symmetric difference} (SD). That is, the number of distinct $A\in 2^N$ such that either $A\in P$ or else
$A\in Q$ but not both. This distance counts the number of non-matched entries in array representations
$\chi_P,\chi_Q:2^N\rightarrow\{0,1\}$, with $\chi_P(A)=1$ if $A\in P$ and 0 otherwise for all $A\in 2^N$ and similarly for $Q$.

Any lattice has a rank function $r(\cdot)$, mapping elements into their level of the Hasse diagram. For the partition lattice,
$r:\mathcal P^N\rightarrow\mathbb Z_+$ is $r(P)=n-|P|$. Given how the rank of subsets appears in (1) above, a further
\textit{rank-based} (RB) partition distance measure is
\begin{equation}
\delta^{RB}(P,Q)=r(P\vee Q)-r(P\wedge Q)=|P\wedge Q|-|P\vee Q|\text ,
\end{equation}
where $P\wedge Q$ is the coarsest partition finer than both $P,Q$ and $P\vee Q$ is the finest partition coarser than both
$P,Q$. Note that any block $A\in P,Q$ of both partitions is also a block of both $P\wedge Q$ and $P\vee Q$, and vice versa.

These simple attempts to parallel (1) already provide two further partition distance measures to be compared with
partition-distance $D(\cdot,\cdot)$. This is done hereafter firstly in terms of the behavior on modular partitions, and secondly
in terms of some suitable classifiers.


\subsection{Distances between modular partitions}
Modular elements and modular pairs (of elements) are very important for comprehending geometric lattices \cite{Aigner79,Stanley71,Stern99},
making it useful to observe the behavior of distance measures over pairs of modular partitions (not to be confused with \textit{modular pairs} of
partitions).

The bottom and top elements of partition lattice $(\mathcal P^N,\wedge,\vee)$ are, respectively, $P_{\bot}=\{\{1\},\ldots ,\{n\}\}$
and $P^{\top}=\{N\}$. Both are modular elements of the lattice. The collection of all modular partitions is
\begin{equation*}
\mathcal P^N_{mod}=\left\{\{A\}\cup P^{A^c}_{\bot}:\emptyset\subset A\subseteq N\right\}\text ,
\end{equation*}
with $P^{\emptyset}=\emptyset$ for all $P\in\mathcal P^N$, and where $\{A\}\cup P^{A^c}_{\bot}$ is the partition with all $i\in A$
in a common block and every $j\in A^c$ in a 1-cardinal block. Note that all the $n$ atoms of $2^N$ (that is, all elements $i\in N$)
collapse into a unique modular partition, which is the bottom one $P_{\bot}$. Hence, $|\mathcal P^N_{mod}|=2^n-n$.

When restricted to $\mathcal P^N_{mod}\times\mathcal P^N_{mod}$, partition-distance $D(\cdot,\cdot)$ above behaves as
follows: $D(P^{\top},P_{\bot})=n-1$, while for $\emptyset\subset A\subset N$
\begin{eqnarray*}
D(\{A\}\cup P^{A^c}_{\bot},P_{\bot})&=&|A|-1\text ,\\
D(\{A\}\cup P^{A^c}_{\bot},P^{\top})&=&|A^c|=n-|A|\text ,\\
D(\{A\}\cup P^{A^c}_{\bot},\{A^c\}\cup P^A_{\bot})&=&|A|-1+|A^c|-1=n-2\text .
\end{eqnarray*}
In general, for $\emptyset\subset A,B\subset N$ and $A\neq B\neq A^c$,
\begin{equation*}
D(\{A\}\cup P^{A^c}_{\bot},\{B\}\cup P^{B^c}_{\bot})=n-|A\cap B|-|(A\cup B)^c|\text .
\end{equation*}
This obtains by firstly determining a largest subset $A'\in 2^N$ where $P$ and $Q$ induce the same partition $P^{A'}=Q^{A'}$,
and next counting the cardinality of its complement. In $(P^{\top},P_{\bot})$ the sought largest subset is any $A'$ such that
$|A'|=1$ (any atom of $2^N$). In $(\{A\}\cup P^{A^c}_{\bot},P_{\bot})$ it is any $A'=A^c\cup i$ for some $i\in A$. In
$(\{A\}\cup P^{A^c}_{\bot},P^{\top})$ it is $A'=A$. In $(\{A\}\cup P^{A^c}_{\bot},\{A^c\}\cup P^A_{\bot})$ it is any $A'=\{i,j\}$
such that $i\in A,j\in A^c$. Finally, for the general case
$(\{A\}\cup P^{A^c}_{\bot},\{B\}\cup P^{B^c}_{\bot})$, these two modular partitions are seen to coincide when restricted to largest
subset $A'=(A\cap B)\cup(A^c\cap B^c)=(A\cap B)\cup(A\cup B)^c\neq\emptyset$ for all $A,B\in 2^N,B\neq A^c$. It may be noted that
$D(\{A\}\cup P^{A^c}_{\bot},\{B\}\cup P^{B^c}_{\bot})=d(A,B)$ as given by (1).   

The restriction of distance $\delta^{SD}$ above to pairs of modular partitions is $\delta^{SD}(P^{\top},P_{\bot})=n+1$ while
$\emptyset\subset A\subset N$ yields
\begin{eqnarray*}
\delta^{SD}(\{A\}\cup P^{A^c}_{\bot},P_{\bot})&=&|A|+1\text ,\\
\delta^{SD}(\{A\}\cup P^{A^c}_{\bot},P^{\top})&=&n-|A|+2\text ,\\
\delta^{SD}(\{A\}\cup P^{A^c}_{\bot},\{A^c\}\cup P^A_{\bot})&=&n+2\text .
\end{eqnarray*}
In general, $\emptyset\subset A,B\subset N$ and $A\neq B\neq A^c$ yield
\begin{eqnarray*}
\delta^{SD}(\{A\}\cup P^{A^c}_{\bot},\{B\}\cup P^{B^c}_{\bot})&=&|A^c|+1+|B^c|+1-2|A^c\cap B^c|\\
&=&2(n+1)-(|A|+|B|)-2|(A\cup B)^c|\text .
\end{eqnarray*}

Concerning $\delta^{RB}$, firstly consider that for $A,B\in 2^N$ the meet and join of the two corresponding modular partitions are
\begin{equation*}
\{A\}\cup P^{A^c}_{\bot}\wedge\{B\}\cup P^{B^c}_{\bot}=\{A\cap B\}\cup P^{(A\cap B)^c}_{\bot}\text ,
\end{equation*}
with possibly $A\cap B=\emptyset$, and
\begin{equation*}
\{A\}\cup P^{A^c}_{\bot}\vee\{B\}\cup P^{B^c}_{\bot}=\left\{\begin{array}{c}\{A\cup B\}\cup P_{\bot}^{(A\cup B)^c}
\text{ if }B\cap A\neq\emptyset\text ,\\
\{A,B\}\cup P^{(A\cup B)^c}_{\bot}\text{ if }A\cap B=\emptyset\text .\end{array}\right.
\end{equation*}
Accordingly, the restriction of distance $\delta^{RB}$ above to pairs of modular partitions is
$\delta^{RB}(P^{\top},P_{\bot})=n-1=D(P^{\top},P_{\bot})$ while $\emptyset\subset A\subset N$ yields
\begin{eqnarray*}
\delta^{RB}(\{A\}\cup P^{A^c}_{\bot},P_{\bot})&=&|A|-1=D(\{A\}\cup P^{A^c}_{\bot},P_{\bot})\text ,\\
\delta^{RB}(\{A\}\cup P^{A^c}_{\bot},P^{\top})&=&n-|A|=D(\{A\}\cup P^{A^c}_{\bot},P^{\top})\text ,\\
\delta^{RB}(\{A\}\cup P^{A^c}_{\bot},\{A^c\}\cup P^A_{\bot})&=&n-2=D(\{A\}\cup P^{A^c}_{\bot},\{A^c\}\cup P^A_{\bot})\text .
\end{eqnarray*}
In general, $\emptyset\subset A,B\subset N$ and $A\neq B\neq A^c$ yield
\begin{equation*}
\delta^{RB}(\{A\}\cup P^{A^c}_{\bot},\{B\}\cup P^{B^c}_{\bot})=
\left\{\begin{array}{c}|A|+|B|-2\text{ if }A\cap B=\emptyset\text{, while}\\
|A\cup B|-|A\cap B|\text{ if }A\cap B\neq\emptyset\text .\end{array}\right .
\end{equation*}

Despite the common range, $\delta^{RB}(\cdot,\cdot)$ and $D(\cdot,\cdot)$ do not coincide even when restricted to modular partitions
(see case $A\cap B\neq\emptyset$ above). Great differences may be checked to arise over pairs of partitions
$P,Q$ where one \textit{covers} the other, denoted $P>^*Q$, meaning $P>Q$ and there is no $P'\in\mathcal P^N$ such that
$P>P'>Q$. For subsets, $A\supset^*B$ when $A=B\cup i$ for some $i\in B^c$.

Perhaps these behaviors enable to figure the functioning of the three distance measures, but still the number
$|\mathcal P^N\backslash\mathcal P^N_{mod}|=\mathcal B_n-2^n+n$ of non-modular partitions is huge for relevant $n$, where
$\mathcal B_n$ is the ($n$-th \textit{Bell}) number of partitions of a $n$-set \cite{Graham++94,Rota64B}. Accordingly, some
general tools for comparison are now provided.


\section{Partition distance measures: classifiers}
Complementation \cite{Aigner79,Stanley71,Stern99} in the partition and subset lattices acts in very different manners: while every subset
has a unique complement (see above), every partition $P\in\mathcal P^N$ has at least one complement ($n>1$), but non-modular ones have many.
They are all those $P'\in\mathcal P^N$ such that $P\wedge P'=P_{\bot}$ as well as $P\vee P'=P^{\top}$. For every partition $P\in\mathcal P^N$,
let $\mathcal P^N_{P^c}$ contain all its complements. 

A partition distance measure $\delta:\mathcal P^N\times\mathcal P^N\rightarrow\mathbb Z_+$ should satisfy
\begin{itemize}
\item $\delta(P,Q)=0\Leftrightarrow P=Q$ for all $P,Q\in\mathcal P^N$ (\textit{antisymmetry}),
\end{itemize}
while further conditions may be the following:  
\begin{enumerate}
\item $\underset{P,Q\in\mathcal P^N}{\max}\delta(P,Q)=f(n)$ ($f$\textit{-maximality}),
\begin{itemize}
\item $f(n+1)>f(n)$ (strong $f$\textit{-monotonicity}),
\item $f(n+2)+f(n)>2f(n+1)$ for all $n\in\mathbb N$ (strong $f$\textit{-convexity}),
\end{itemize}
\item $\underset{P,Q\in\mathcal P^N}{\max}\delta(P,Q)=\underset{P,Q\in\mathcal P^N_{mod}}{\max}\delta(P,Q)$ \textit{(mod-maximality)},
\item $\underset{P,Q\in\mathcal P^N}{\max}\delta(P,Q)=\delta(P_{\bot},P^{\top})$ ($\bot\top$\textit{-maximality}),
\item $\underset{P,Q\in\mathcal P^N}{\max}\delta(P,Q)=\delta(P,Q)$ for all $P\in\mathcal P^N,Q\in\mathcal P^N_{P^c}$ (\textit{co-maximality}),
\item $\delta(P\wedge Q,P\vee Q)\geq\delta(P,Q)$ for all $P,Q\in\mathcal P^N$ (\textit{super-modularity}),
\item $\delta(P\wedge Q,P\vee Q)\leq\delta(P,Q)$ for all $P,Q\in\mathcal P^N$ (\textit{sub-modularity}),
\item $\delta(P\wedge Q,P\vee Q)=\delta(P,Q)$ for all $P,Q\in\mathcal P^N$ (\textit{modularity}).
\end{enumerate}

The preliminary statement is obvious: there is no distance between any partition and itself as well as, conversely, if there is no
distance between two partitions then they coincide (see also \cite[def. 3]{SoberLook06}).

The first condition states that the maximum distance between two partitions of a $n$-set is a function
$f:\mathbb N\rightarrow\mathbb Z_+$ of $n$ only. Then, antisymmetry entails $f(1)=0$, as there is a unique partition of a
singleton. In addition, the first $f(n+1)-f(n)$ and second $f(n+2)-f(n+1)-(f(n+1)-f(n))$ differences may be both strictly positive.

Conditions 2-4 all select a region of the product lattice $\mathcal P^N\times\mathcal P^N$ where the measure has to surely
attain its maximum, without excluding that such a maximum may be also attained elsewhere. Specifically, condition 2 states that
the maximum distance between any two partitions of a $n$-set is the same as that observed as the maximum distance
between any two modular partitions of the set. Condition 3 requires, in addition, that the pair consisting of the bottom and top
partitions is among the maximizers of the distance. Condition 4 requires, in addition, that any pair consisting of a partition and
one of its complements is among the maximizers of the distance. Hence, each entails the preceding one:
$4\Rightarrow3\Rightarrow 2$.

A main observation for discussing conditions 5-7 is that partition distance measures have to act on pairs $P,Q$ that are
incomparable in terms of coarsening $\geqslant$, that is $P\not\geqslant Q\not\geqslant P$ (hence they are excluded from the
incidence algebra of the partition lattice \cite{Aigner79,Stern99}). In this case, it may be important to know if a distance measure
behaves differently depending on whether the two involved partitions are comparable or not. More precisely, the issue is
comparing distance $\delta(P,Q)$ with the most similar distance between partitions that are comparable, namely
$\delta(P\wedge Q,P\vee Q)$. In the Hasse diagram, the left-right distance between incomparable partitions $P,Q$ is replaced
with the up-down distance between $P\vee Q,P\wedge Q$. In this view, a sub-(super-)modular distance measure translates the
idea that by switching from an incomparable pair to the nearest comparable one the distance decreases (increases). More
simply, a distance measure is \textit{modular} when it deals with both comparable and incomparable pairs exactly in the same
manner, being a maximal sub-modular and minimal super-modular one.

Classifiers 5-7 borrow their names from lattice functions $h:X^N\rightarrow\mathbb R$, taking real values on a
lattice $(X^N,\wedge,\vee)$ with meet $\wedge$, join $\vee$ (and, possibly, built upon some finite set $N$ as above). Such functions are
sub-modular when $h(x\vee y)+h(x\wedge y)\leq h(x)+h(y)$ for all pairs $x,y\in X^N$ of lattice elements, and are key tools in
combinatorial theory and optimization
\cite{Aigner79,SubFunOptimization05,GroetschelLovaszSchrijver1988,Stern99}.
Super-modularity obtains when the inequality is reversed. Lattice functions satisfying both sub- and super-modularity are mostly
referred to as modular (or additive or valuations). The literature may be found generally concerned more with modular set
functions rather than modular partition functions; the reason is simple: the only way a function can be a modular in the partition lattice is
by assigning the same constant value to every partition \cite[exercise 12 (ii), p. 195]{Aigner79}. It must be stressed though, that these names borrowed
from lattice functions are here applied, instead, to distance measures. These latter map pairs of lattice elements, while a function
maps lattice elements. Hence, a modular  partition distance measure is reasonable (as  long as it is not built upon a modular
partition function, see below). 

\subsection{Characterization}
Conditions 1-7 apply to any complemented lattice, and thus straightforwardly allow to classify the Hamming distance 
$d(\cdot,\cdot):2^N\times2^N\rightarrow\{0,1,\ldots ,n\}$ between subsets in (1) above: $d(\cdot,\cdot)$ simply satisfies
all conditions apart from strong $f$-convexity, as $f(n)=n$. In this view, RB partition distance measure $\delta^{RB}(\cdot,\cdot)$
defined by (4) above behaves exactly the same as $d(\cdot,\cdot)$, satisfying all conditions apart from strong $f$-convexity, with
$f(n)=n-1$. Conversely, partition-distance $D(\cdot,\cdot)$ and SD distance $\delta^{SD}(\cdot,\cdot)$ (from (2) and (3) above) only
satisfy certain conditions out of 1-7, and appear substantially different from $\delta^{RB}(\cdot,\cdot)$ (apart from the immediate
check that $D(\cdot,\cdot)$ satisfies $f$-maximality and strong $f$-monotonicity, but not strong $f$-convexity,
as $f(n)=n-1$, like $\delta^{RB}(\cdot,\cdot)$).  

\begin{claim}
Partition-distance $D(\cdot,\cdot)$ given by (2) is super-modular:
\begin{equation*}
D(P\vee Q,P\wedge Q)-D(P,Q)\geq 0\text{ for all }P,Q\in\mathcal P^N\text .
\end{equation*}
\end{claim}
\begin{proof}
Partition-distance $D(P,Q)$ is $n-|A|$ where $A\in 2^N$ is a largest subset satisfying $P^A=Q^A$, while partition-distance
$D(P\vee Q,P\wedge Q)$ is $n-|B|$ where $B$ is a largest subset satisfying $(P\vee Q)^B=(P\wedge Q)^B$. What remains to
note is $(P\vee Q)^A\geqslant P^A\geqslant(P\wedge Q)^A\leqslant Q^A\leqslant(P\vee Q)^A$ for all $A\in 2^N$. This means
that for every $A\in 2^N$, if $(P\vee Q)^A=(P\wedge Q)^A$, then $P^A=Q^A$.
\end{proof}
\smallskip

For example, let $N=\{1,2,3,4\}$ and consider partitions $P,Q\in\mathcal P^N$ with $P=\{12|34\}$ and $Q=\{13|24\}$, where $|$ separates blocks.
Then, $P\vee Q=P^{\top}$ and $P\wedge Q=P_{\bot}$, and thus $D(1234,1|2|3|4)=3>2=D(12|34,13|24)$. 
\begin{claim}
Distance measure $\delta^{SD}(\cdot,\cdot)$ given by (3) is super-modular:
\begin{equation*}
\delta^{SD}(P\vee Q,P\wedge Q)-\delta^{SD}(P,Q)\geq 0\text{ for all }P,Q\in\mathcal P^N\text .
\end{equation*}
\end{claim}
\begin{proof}
As $\delta^{SD}(P,Q)$ counts the number of blocks of either $P$ or $Q$ but not both, it must be shown that the way such
blocks are further partitioned in $P\wedge Q$ and merged in $P\vee Q$ yields an overall number of blocks no smaller than
$\delta^{SD}(P,Q)$. In fact, this is evident when considering that the partition lattice is the \textit{polygon matroid}
\cite[theorem 6.23, p. 274]{Aigner79}, and any matroid has a sub-modular rank function \cite[rank axioms 6.14, p. 265]{Aigner79} (see above).
That is, $r(P\vee Q)+r(P\wedge Q)\leq r(P)+r(Q)$ for all $P,Q\in\mathcal P^N$. Then,
\begin{eqnarray*}
n-|P\vee Q|+n-|P\wedge Q|&\leq& n-|P|+n-|Q|\text ,\\
|P\vee Q|+|P\wedge Q|&\geq& |P|+|Q|\text ,\\
|(P\vee Q)\backslash(P\wedge Q)|+|(P\wedge Q)\backslash(P\vee Q)|&\geq& |P\backslash Q|+|Q\backslash P|\text ,\\
\delta^{SD}(P\vee Q,P\wedge Q)&\geq&\delta^{SD}(P,Q)\text ,
\end{eqnarray*}
as $|(P\vee Q)\cap(P\wedge Q)|=|P\cap Q|$.
\end{proof}
\smallskip

For example, let $N=\{1,2,3,4,5,6,7\}$ and consider partitions $P,Q\in\mathcal P^N$ with $P=\{12|34|567\}$ and $Q=\{12|35|467\}$.
Then, $P\vee Q=\{12|34567\}$ and $P\wedge Q=\{12|3|4|5|67\}$, and thus $\delta^{SD}(12|34567,12|3|4|5|67)=5$ while
$\delta^{SD}(12|34|567,12|35|467)=4$. On the other hand, $P'=\{12|34|56|7\}$ and $Q'=\{12|3|45|67\}$ yield $P'\vee Q'=\{12|34567\}$
and $P'\wedge Q'=\{12|3|4|5|6|7\}$, and thus $\delta^{SD}(P'\vee Q',P'\wedge Q')=6=\delta^{SD}(P',Q')$.

\begin{claim}
Neither $D(\cdot,\cdot)$ nor $\delta^{SD}(\cdot,\cdot)$ satisfy co-maximality.
\end{claim}
\begin{proof}
Concerning $D(\cdot,\cdot)$, the proof consists in providing a pair of complements $P,Q$ between which partition-distance
$D(P,Q)$ is strictly less than the maximum $n-1$. To this end, let $n$ odd and sufficiently large. Consider $P=\{A,B,\{i\}\}$ and
$Q=\{\{i,j,j'\}\cup P^{\{i,j,j'\}^c}_{\bot}\}$ with $|A|=|B|=\frac{n-1}{2}$ as well as $j\in A,j'\in B$. Then, $P\wedge Q=P_{\bot}$
as well as $P\vee Q=P^{\top}$, and yet $D(P,Q)=n-3$, in that both $P$ and $Q$ induce the same partition of any 3-cardinal subset of
the form $\{i,l,l'\}$ such that $l\in A\backslash j,l'\in B\backslash j'$.

Concerning $\delta^{SD}(\cdot,\cdot)$, a stronger result is actually obtained, namely that this measure does not even satisfy
$\bot\top$-maximality. To see this, again let $n$ odd and sufficiently large; in particular, $\frac{n+1}{4}\in\mathbb N$. Let
$P=P^A\cup P^{A^c}_{\bot}$ and $Q=Q^B\cup P_{\bot}^{B^c}$ with $|A\cap B|=1$ as well as $|P^A|=\frac{n+1}{4}=|Q^B|$. In words,
both $P,Q$ have only 2- and 1-cardinal blocks, and the same numbers $\frac{n+1}{4}$ and $\frac{n-1}{2}$ of blocks for each of
these two cardinalities, respectively. In addition, only one element $i\in N$ (of the set being partitioned) is included in some
2-cardinal block both in $P$ and in $Q$, that is $\{i\}=A\cap B$ (while all other elements $j\in N\backslash i$ are in a 2-cardinal
block of $P$ and in a 1-cardinal block of $Q$, or vice versa). Then,
$\delta^{SD}(P,Q)=2\frac{n+1}{4}+2\frac{n-1}{2}=\frac{3n-1}{2}>n+1=\delta^{SD}(P^{\top},P_{\bot})$.
\end{proof}

\begin{claim}
For all $P,Q\in\mathcal P^N$, if $P\geqslant Q$, then
\begin{equation}
\delta^{SD}(P,Q)=\delta^{RB}(P,Q)+2|P\backslash Q|=2|Q\backslash P|-\delta^{RB}(P,Q)\text .
\end{equation}
\end{claim}
\begin{proof}
If $P\geqslant Q$, then $\delta^{RB}(P,Q)=r(P\vee Q)-r(P\wedge Q)=$
\begin{eqnarray*}
&=&n-|P\backslash Q|-|P\cap Q|-\left(n-|Q\backslash P|-|P\cap Q|\right)\\
&=&|Q\backslash P|-|P\backslash Q|=\delta^{SD}(P,Q)-2|P\backslash Q|=2|Q\backslash P|-\delta^{SD}(P,Q)
\end{eqnarray*}
as wanted.
\end{proof}
\begin{claim}
For all $P,Q\in\mathcal P^N$,
\begin{equation*}
\delta^{SD}(P,Q)=2(n-|P\cap Q|)-(r(P)+r(Q))\text.
\end{equation*}
\end{claim}
\begin{proof}
Simply by substitution:
\begin{eqnarray*}
\delta^{SD}(P,Q)&=&|P\backslash Q|+|Q\backslash P|\\
&=&|P\backslash Q|+|P\cap Q|+|Q\backslash P|+|P\cap Q|-2|P\cap Q|\\
&=&n-r(P)+n-r(Q)-2|P\cap Q|
\end{eqnarray*}
as wanted.
\end{proof}
\smallskip

It seems important recognizing how the meet $\wedge$ and join $\vee$ operators of the partition lattice are used in different
manners by the RB and SD distance measures. Both perform a count based on the blocks of either one but not both the involved
partitions $P,Q$. These are precisely the blocks disjoined by $\wedge$ and adjoined by $\vee$. Yet, RB distance counts the
number of blocks resulting from the join and subtracts it from the number of blocks resulting from the meet. Of course, blocks of
both the meet and the join vanish through the subtraction. Conversely, SD counts the whole number of blocks of either one but
not both partitions $P,Q$. Hence, when these latter are comparable in terms of coarsening, say $P\geqslant Q$, condition (5) is plain.

Although the RB distance behaves exactly the same as the Hamming distance between subsets according to classifiers 1-7 above,
still the former does not seem to properly translate the latter in terms of partitions. In particular, as both
$D(\cdot,\cdot),\delta^{SD}(\cdot,\cdot)$ are super-modular and do not satisfy co-maximality, these latter two measures are actually
preferable over $\delta^{RB}(\cdot,\cdot)$. The reason for this, roughly speaking, is that the subset and partition lattices are very
different, and RB distance simply ignores such differences.

Focus on super-modularity first. With their two Hasse diagrams in mind, consider that there are $\mathcal B_n-2^n$ more partitions than
subsets of a $n$-set, and such a gap grows dramatically fast as $n$ increases. Yet, partitions are compressed into $n$ levels, one less
than subsets. There are $\binom{n}{k}$ distinct $k$-subsets of a $n$-set, $0\leq k\leq n$, while there are
$\mathcal S_{n,k}=\sum_{0\leq m\leq k}(-1)^{k-m}\binom{k}{m}\frac{m^n}{k!}$ distinct ways to partition a $n$-set into $k$ blocks,
$0<k\leq n$, where $\mathcal S_{n,k}$ are the \textit{Stirling numbers of the second kind} \cite[p. 265]{Graham++94} or cardinalities
of levels $n-k,0<k\leq n$ of the partition lattice.

While moving down-upward in the Hasse diagram, in both lattices the cardinality of levels firstly increases, reaching a maximum, and then
decreases. Yet, in the subset lattice such a maximum is always reached at levels $\{\lfloor\frac{n}{2}\rfloor,\lceil\frac{n}{2}\rceil\}$
whenever they differ (and at level $\frac{n}{2}\in\mathbb N$ otherwise), and the preceding ascent is exactly the same as the following
descent. No such a regular behavior is displayed by partitions, as the upper part of the Hasse diagram is much more populated than the
lower one. In fact, the maximum density attains quite above the half level, making the preceding ascent slow and the following descent
fast \cite[pp. 91-92]{Aigner79}, \cite{Canfield+95}.

All this leads to conclude that when up-down distances between partitions are replaced with left-right ones (see above on
sub/super-modularity), a kind of quantitative expansion occurs with respect to the subset lattice, in that there are many more pairs of
incomparable partitions than pairs of incomparable subsets, simply because most level sets are massively more populated in the partition
lattice rather than in the subset one. Given such an expansion, any distance measure such as RB  given in (4), that compares $P,Q$ by
taking into account, in some fashion, the whole segment (or sub-lattice) $[P\wedge Q,P\vee Q]$, becomes next forced to also take into
account, in the same fashion, all the differences between partitions into such a segment. Conversely, SD and partition-distance
are not under such a forcing, and thus can adapt their behavior to a proper subset of the segment.

As for complementarity, it is crucial noting (again) that non-modular partitions have many complements, and these latter differ in terms of both
the number and the cardinalities of blocks \cite{Stanley71}. Accordingly, asking a distance measure to attain its maximum on every pair
of complements is reasonable in the subset lattice but becomes far too binding when dealing with partitions. This is the second reason why
$\delta^{RB}(\cdot,\cdot)$ is less desirable than $\delta^{SD}(\cdot,\cdot),D(\cdot,\cdot)$.

Finally, among these latter two, SD distance is better because it takes much more values than partition-distance. More precisely, as
partitions may differ in a number of distinct ways that greatly exceeds $n-2=|\{1,\ldots ,n-1\}|$, there are many differences between
partitions which are substantially diverse while still being mapped by $D(\cdot,\cdot)$ into a same integer between 1 and $n-1$.
Conversely, SD distance is able to recognize that such differences are diverse, and thus maps them into distinct (integer-valued)
distances. In this view, an even better solution to the problem of quantitatively discriminating between differences
that are factually diverse is proposed in the sequel. Still, a super-modular distance measure
$\delta^{RB}_+:\mathcal P^N\times\mathcal P^N\rightarrow\mathbb Z_+$ not satisfying co-maximality may be constructed even
by resorting simply to the rank:
\begin{eqnarray}
\delta^{RB}_+ (P,Q)&=&r(P)+r(Q)-2r(P\wedge Q)\\
&=&\delta^{RB}(P,Q)+\big(r(P)+r(Q)-\big(r(P\vee Q)+r(P\wedge Q)\big)\big)\text .
\end{eqnarray}
This distance is super-modular precisely because the rank is a sub-modular partition function, and coincides with $\delta^{RB}(P,Q)$
if and only if $P,Q$ is a modular pair \cite{Stanley71}, that is, if and only if $r(P\vee Q)+r(P\wedge Q)=r(P)+r(Q)$. It is also easily
checked that $\delta^{RB}_+(\cdot,\cdot)$ does not satisfy co-maximality.

Elementary though it is, one important observation is now the following: the rank is a monotone lattice function through which RB distance
quantifies differences between lattice elements. This may be generalized: once endowed with a monotone lattice function $h$ on $X^N$,
that is $h(x)\geq h(y)$ for all $x,y\in X^N,x\geqslant y$, differences between elements $x,y\in X^N$ can be promptly quantified by
distance $\delta(x,y)=h(x\vee y)-h(x\wedge y)$, which is evidently modular by construction. Then, RB measure
uses the rank, but any other monotone partition function works. An alternative one is hereafter. 


\section{Atoms and the size}
Apart from the bottom and top, among the remaining $2^n-n-2$ elements of $\mathcal P^N_{mod}$ (let $n>2$) there are $\binom{n}{2}$ modular
partitions playing a crucial role in what follows. They are the atoms of the partition lattice, consisting each of $n-1$ blocks, one being
2-cardinal and all remaining ones being 1-cardinal. For $1\leq i<j\leq n$, denote by $[ij]=\{i,j\}\cup P_{\bot}^{N\backslash\{i,j\}}$ the
atom whose unique 2-cardinal block is $\{i,j\}$, with $\mathcal P^N_1=\{[ij]:1\leq i<j\leq n\}\subseteq\mathcal P^N_{mod}$ containing
all $\binom{n}{2}$ such atoms\footnote{Note that $n=1$ yields $\mathcal P^N_1=\emptyset$, while $n=2$ yields $\mathcal P^N_1=\{P^{\top}\}$
as well as $n=3$ yields $\mathcal P^N_1=\mathcal P^N_{mod}\backslash\{P_{\bot},P^{\top}\}$. Also, $\mathcal P^N_{mod}=\mathcal P^N$ for $n\leq 3$.}.

The focus now turns on representing partitions $P$ as strings $I_P\in\{0,1\}^{\binom{n}{2}}$. For every partition $P\in\mathcal P^N$,
consider the array representation or indicator function $I_P:\mathcal P^N_1\rightarrow\{0,1\}$ defined
by $I_P([ij])=1$ if $P\geqslant[ij]$ and $I_P([ij])=0$ if $P\not\geqslant[ij]$. This is clearly the analog of the characteristic function $\chi_A$
for subsets $A\in 2^N$. Yet, a fundamental distinction must be immediately emphasized: while $\chi$ is a bijection, in that
$\left\{\chi_A:A\in 2^N\right\}=\{0,1\}^n$, the partition indicator function does not reach every vertex of the $\binom{n}{2}$-dimensional unit
hypercube, as $\left\{I_P:P\in\mathcal P^N\right\}\subset\{0,1\}^{\binom{n}{2}}$. This redundancy is due to linear dependence,
characterizing geometric lattices in general \cite{Aigner79,Stern99,Whitney35}.

The partition indicator function $I:\mathcal P^N\rightarrow\{0,1\}^{\binom{n}{2}}$, with $I(P)=I_P$, enables to introduce the size
$s:\mathcal P^N\rightarrow\mathbb Z_+$, firstly appearing in \cite{RoxyJean04} as the analog (in a sense made clearer shortly) of the cardinality
of subsets. The size $s(P)=s^P$ is the number of atoms finer than $P$, that is,
\begin{equation*}
s^P=\left|\left\{[ij]\in\mathcal P^N_1:P\geqslant[ij]\right\}\right|=\sum_{[ij]\in\mathcal P^N_1}I_P([ij])\text .  
\end{equation*}

The size maps partitions of a $n$-set into the first $\binom{n}{2}+1$ positive integers, but many of these latter are left out. That is,
there are naturals $s<\binom{n}{2}$, such that $s\neq s^P$ for all $P\in\mathcal P^N$. The available sizes for partitions of a $n$-set, $n\leq 7$, are as follows:
\begin{eqnarray*}
|N|=n&\rightarrow&\{s^P:P\in\mathcal P^N\}\text{ (available sizes)}\\
1&\rightarrow&\{0\}\\
2&\rightarrow&\{0,1\}\\
3&\rightarrow&\{0,1,3\}\\
4&\rightarrow&\{0,1,2,3,6\}\\
5&\rightarrow&\{0,1,2,3,4,6,10\}\\
6&\rightarrow&\{0,1,2,3,4,6,7,10,15\}\\
7&\rightarrow&\{0,1,2,3,4,5,6,7,9,10,11,15,21\}\text .
\end{eqnarray*}

On the enumerative side, the size obtains from the \textit{class} $c:\mathcal P^N\rightarrow\mathbb Z^n_+$, where
$c(P)=c^P=(c^P_1,\ldots c^P_n)$ with $c^P_k=|\{A\in P:|A|=k\}|$ counting the number of $k$-cardinal blocks of $P$,
for $1\leq k\leq n$. Then,
\begin{equation*}
s^P=\sum_{1\leq k\leq n}c^P_k\binom{k}{2}=\sum_{A\in P}\binom{|A|}{2}\text .  
\end{equation*}

\begin{claim}
The size is a strictly monotone partition function:
\begin{equation*}
s^P>s^Q\text{ for all }P,Q\in\mathcal P^N\text{ such that }P>Q\text .
\end{equation*}
\end{claim}

\begin{proof}
If $P>Q$, then at least one block $A\in P$ is the union of some blocks $B_1,\ldots ,B_m\in Q,m\geq 2$. Merging any two such $B,B'$ increases the size by
\begin{equation*}
\binom{|B|+|B'|}{2}-\left(\binom{|B|}{2}+\binom{|B'|}{2}\right)=|B||B'|\text ,
\end{equation*}
which is strictly positive as blocks are non-empty.
\end{proof}

\begin{claim}
The size is a super-modular partition function:
\begin{equation*}
s^{P\vee Q}+s^{P\wedge Q}\geq s^P+s^Q\text{ for all }P,Q\in\mathcal P^N\text .
\end{equation*}
\end{claim}

\begin{proof}
If the two partitions are comparable, say $P\geqslant Q$, then $P=P\vee Q$ and $Q=P\wedge Q$, which makes the statement satisfied with equality.
Otherwise, $P\not\geqslant Q\not\geqslant P$ entails that there are two maximal chains of partitions, one of which meets $P\wedge Q$ and $P$ as
well as $P\vee Q$, while the other meets $P\wedge Q$ and $Q$ as well as $P\vee Q$. Focusing on the relevant part or segment\footnote{A chain,
possibly maximal, is a totally ordered sub-lattice, and thus has segments.} of the former maximal chain, there are
$\hat P_{\hat r}>^*\cdots>^*\hat P_1>^*\hat P_0$, with $\hat r=r(P\vee Q)-r(P\wedge Q)$, such that $\hat P_0=P\wedge Q$ and $\hat P_{\hat r}=P\vee Q$
as well as $\hat P_{k_P}=P$ for some $k_P,0<k_P<\hat r$. Similarly, focusing on the relevant segment of the latter maximal chain\footnote
{The length $\hat r$ is the same for the two segments.}, there are $\hat Q_{\hat r}>^*\cdots>^*\hat Q_1>^*\hat Q_0$
such that $\hat Q_0=P\wedge Q$ and $\hat Q_{\hat r}=P\vee Q$ as well as $\hat Q_{k_Q}=Q$ for some $k_Q,0<k_Q<\hat r$. Note that if $r(P)=r(Q)$,
then $k_P=k_Q$. 

The count $s^{P\vee Q}+s^{P\wedge Q}-(s^P+s^Q)$ may be performed by focusing on each level of the two segments. The fact is that most atoms finer
than $P\vee Q$ are $\geqslant$-incomparable with respect to both $P$ and $Q$. Atoms $[ij]\leqslant P\wedge Q$ may be ignored because
they are counted in the size of all the four involved partitions $P,Q,P\wedge Q,P\vee Q$. As for the remaining ones, observe that
\begin{equation*}
\Big\{[ij]\in\mathcal P^N_1:P\geqslant[ij]\not\leqslant P\wedge Q\Big\}\bigcap
\Big\{[ij]\in\mathcal P^N_1:Q\geqslant[ij]\not\leqslant P\wedge Q\Big\}=\emptyset\text .
\end{equation*}
To see this, assume an atom $[ij]\not\leqslant P\wedge Q$ satisfies $P\geqslant[ij]\leqslant Q$. Then, $(P\wedge Q)\vee [ij]$, and \textit{not}
$P\wedge Q$, would be the coarsest partition finer than both $P,Q$. In particular,
$(P\wedge Q)\not\geqslant[ij]\Rightarrow((P\wedge Q)\vee[ij])>^*(P\wedge Q)$.
  
Consider going from $P\wedge Q$ to $P\vee Q$ through the Hasse diagram \textit{twice}, initially endowed with all atoms finer than $P\vee Q$ apart
from those also finer than $P\wedge Q$. The first route is through segment $\hat P_0,\ldots,\hat P_{\hat r}$ of the former maximal chain, with
the following constraint: at each partition reached up to $\hat P_{k_P}=P$ inclusive, all atoms finer than the current partition but not also finer
than the preceding one must be left there in order to proceed. The second route starts with only the \textit{residual} atoms and is through segment
$\hat Q_0,\ldots,\hat Q_{\hat r}$ of the latter maximal chain. Again, up to $\hat Q_{k_Q}=Q$ inclusive at each reached level all atoms finer than
the current partition but not also finer than the preceding one must be left there in order to proceed. Given the above empty intersection, it is
not possible that an atom is needed twice for proceeding, and at the end of the second route there still remains a non-empty (and large, in general)
collection of atoms, namely all those for reaching $P\vee Q$ from either $P$ or $Q$.   
\end{proof}
\smallskip

From a final perspective, consider that any subset has a unique representation as a join of atoms $i\in N$ of the subset lattice, while linear
dependence makes partitions have, in general, many representations as a join of atoms. Most of them are redundant, in that removing some atom(s)
from the join leaves the represented partition unchanged. In fact, any partition has a unique maximal or largest representation as a join of atoms.
The size counts precisely the cardinality of this largest representation. 


\section{The indicator-Hamming distance measure}
The size enables to introduce two novel partition distances. For reasons immediately explained hereafter, they may be referred to as
follows:
\begin{itemize}
\item the \textit{indicator-Hamming} distance $\delta^{IH}:\mathcal P^N\times\mathcal P^N\rightarrow\mathbb Z_+$ defined by
\begin{equation}
\delta^{IH}(P,Q)=\sum_{[ij]\in\mathcal P^N_1}\Big(I_P([ij])-I_Q([ij])\Big)^2=s^P+s^Q-2s^{P\wedge Q}\text ,
\end{equation}
\item the \textit{size-based} distance $\delta^{SB}:\mathcal P^N\times\mathcal P^N\rightarrow\mathbb Z_+$ defined by
\begin{equation}
\delta^{SB}(P,Q)=\sum_{[ij]\in\mathcal P^N_1}\Big(I_{P\vee Q}([ij])-I_{P\wedge Q}([ij])\Big)=s^{P\vee Q}-s^{P\wedge Q}\text .
\end{equation}

\end{itemize}
Just like the Hamming distance between subsets uses their symmetric difference as the (counting) measure, the IH distance simply counts
the number of non-matched entries $I_P([ij])\neq I_Q([ij])$ for $1\leq i<j\leq n$ between the two array representations $I_P,I_Q$ of any
two partitions $P,Q$ (as $I_P([ij])-I_Q([ij])\in\{-1,0,1\}$). This means counting the number of atoms finer than either one of the two
partitions but not both, which is exactly what the Hamming distance between subsets does in (1) above. Accordingly, this IH measure is
here conceived as the faithful reproduction of the (cardinality of the) symmetric difference between subsets. In terms of the above
classifiers 1-7, its behavior will shortly appear rather different when compared to the Hamming distance between subsets. In fact, as explained
above, the partition and subset lattices display great differences.

Much more roughly, SB distance counts the number or atoms finer than the meet $P\wedge Q$ and subtracts it from the
number of atoms finer that the join $P\vee Q$. It is immediate noting that the two measures SB and IH coincide on pairs of comparable partitions:
if (say) $P>Q$, then $P\vee Q=P,Q=P\wedge Q$. More generally, these two distances coincide on all and only those pairs $P,Q\in\mathcal P^N$
(possibly $P\not\geqslant Q\not\geqslant P$) where the size function satisfies $s^{P\vee Q}+s^{P\wedge Q}=s^P+s^Q$. It may be checked that this
attains only on modular pairs \cite{Stanley71}, that is,
\begin{equation*}
r(P)+r(Q)=r(P\vee Q)+r(P\wedge Q)\Leftrightarrow s^P+s^Q=s^{P\vee Q}+s^{P\wedge Q}
\end{equation*}
for all $P,Q\in\mathcal P^N$. In this respect, IH distance transforms SB distance similarly to how $\delta^{RB}_+(\cdot,\cdot)$ transforms
$\delta^{RB}(\cdot,\cdot)$ (see (4),(6) and (7) above).

It is mostly important observing that a main distinction between the IH and SB distance measures relies in their ranges (or images \cite[p. 5]{Aigner79}),
that do not coincide, being one a proper subset of the other. The range of the size-based distance contains only certain positive differences
between some available sizes of partitions (see above\footnote{The number of available sizes for partitions of a $n$-set exceeds $n$ for $n>3$;
in fact, as soon as $n>3$ non-modular elements start appearing.}). In addition to these values, attained all the same on modular pairs $P,Q$,
IH distance has a variety of further positive integers in its range. This is evident from super-modularity of the size function, and provides
the needed granularity and local flexibility when quantifying differences between incomparable partitions.

Both measures satisfy $f$-maximality with $f(n)=\binom{n}{2}$, and hence both strong $f$-monotonicity and strong $f$-convexity hold, in that
$\binom{n+1}{2}-\binom{n}{2}=\frac{n}{2}+1$ as well as $\binom{n+2}{2}+\binom{n}{2}-2\binom{n+1}{2}=1$.

Both measures satisfy $\bot\top$-maximality, and thus mod-maximality, but the SB one also satisfies co-maximality, while the IH one does not, In
fact, in (8) the join $P\vee Q$ of the two partitions does not even compare.

By construction, the SB measure is modular, while the IH one is super-modular, precisely because the size is a super-modular partition function:
\begin{eqnarray*}
\delta^{IH}(P\vee Q,P\wedge Q)-\delta^{IH}(P,Q)&=&s^{P\vee Q}+s^{P\wedge Q}-2s^{P\wedge Q}+\\
&-&\big(s^P+s^Q-2s^{P\wedge Q}\big)\\
&=&s^{P\vee Q}+s^{P\wedge Q}-\big(s^P+s^Q\big)\geq 0
\end{eqnarray*}
from above. In fact, SB distance is the minimal modular distance no smaller than IH distance over all pairs of partitions.  

SB distance restricted to $\mathcal P^N_{mod}\times\mathcal P^N_{mod}$ is $\delta^{SB}(P^{\top},P_{\bot})=\binom{n}{2}$, while
for $\emptyset\subset A\subset N$
\begin{eqnarray*}
\delta^{SB}(\{A\}\cup P^{A^c}_{\bot},P_{\bot})&=&\binom{|A|}{2}\text ,\\
\delta^{SB}(\{A\}\cup P^{A^c}_{\bot},P^{\top})&=&\binom{n}{2}-\binom{|A|}{2}\text ,\\
\delta^{SB}(\{A\}\cup P^{A^c}_{\bot},\{A^c\}\cup P^A_{\bot})&=&\binom{|A|}{2}+\binom{n-|A|}{2}\text .
\end{eqnarray*}
Case $\emptyset\subset A,B\subset N$ and $A\neq B\neq A^c$ yields
\begin{equation*}
\delta^{SB}(\{A\}\cup P^{A^c}_{\bot},\{B\}\cup P^{B^c}_{\bot})=\binom{|A\cup B|}{2}-\binom{|A\cap B|}{2}\text ,
\end{equation*}
which reduces to $\binom{|A|}{2}+\binom{|B|}{2}$ whenever $A\cap B=\emptyset$ (as $\binom{0}{2}=\binom{1}{2}=0$).

IH distance restricted to $\mathcal P^N_{mod}\times\mathcal P^N_{mod}$ is $\delta^{IH}(P^{\top},P_{\bot})=\binom{n}{2}$, while
for $\emptyset\subset A\subset N$
\begin{eqnarray*}
\delta^{IH}(\{A\}\cup P^{A^c}_{\bot},P_{\bot})&=&\binom{|A|}{2}\text ,\\
\delta^{IH}(\{A\}\cup P^{A^c}_{\bot},P^{\top})&=&\binom{n}{2}-\binom{|A|}{2}\text ,\\
\delta^{IH}(\{A\}\cup P^{A^c}_{\bot},\{A^c\}\cup P^A_{\bot})&=&\binom{|A|}{2}+\binom{n-|A|}{2}\text .
\end{eqnarray*}
Case $\emptyset\subset A,B\subset N$ and $A\neq B\neq A^c$ yields
\begin{equation*}
\delta^{IH}(\{A\}\cup P^{A^c}_{\bot},\{B\}\cup P^{B^c}_{\bot})=\binom{|A|}{2}+\binom{|B|}{2}-2\binom{|A\cap B|}{2}\text .
\end{equation*}

Even when restricted to the $2^n-n$ modular partitions, these two distance measures still display different behavior in most cases of
incomparability.


\section{A comparison through bounding}
This section compares partition-distance $D(\cdot,\cdot)$ and indicator-Hamming distance $\delta^{IH}(\cdot,\cdot)$ with the intent to figure how many
different values the latter may take for every (non-trivial) value of the former. In fact, for $k$  such that $0<k<n$ ($k=0$ is indeed trivial), the concern is with
the maximum and minimum value taken by $\delta^{IH}(\cdot,\cdot)$ while ranging over all pairs $P,Q\in\mathcal P^N$ satisfying $D(P,Q)=k$. To this end, the following
result is important in that it shows that looking at a largest subset $A\subseteq N$ where any two partitions $P,Q\in\mathcal P^N$ coincide is equivalent to looking
at the largest collection of atoms that are finer than both. 

\begin{claim}
If $A\in 2^N$ is a maximal subset where $P^A=Q^A$, then $s^{P^A}=s^{P\wedge Q}$.
\end{claim}
Note that $P^A\in\mathcal P^A$ while $P\wedge Q\in\mathcal P^N$, but still the size of any partition is a positive integer, and thus the sizes of two partitions are
comparable even when these latter are elements of distinct lattices. In fact,  $\mathcal P^A$ is equivalent to segment
$[P_{\bot},\{A\}\cup P^{A^c}_{\bot}]\subset\mathcal P^N$ (see above).
\smallskip
  
\begin{proof}
If $A=N$, then $P=Q$ and there is nothing to show. Assume $A\subset N$. Then, $P>P^A\cup P_{\bot}^{A^c}$ as well as $Q>P^A\cup P_{\bot}^{A^c}$, entailing that both
$P$ and $Q$ obtain by joining $P^A\cup P^{A^c}_{\bot}$ with atoms $[ij]\not\leqslant P^A\cup P^{A^c}_{\bot}$ as follows
\begin{eqnarray*}
P&=&(P^A\cup P^{A^c}_{\bot})\vee[ij]_1\vee\cdots\vee[ij]_{m_P}\text ,\\
Q&=&(P^A\cup P^{A^c}_{\bot})\vee[ij]'_1\vee\cdots\vee[ij]'_{m_Q}\text .
\end{eqnarray*}
These collections $\{[ij]_1,\ldots ,[ij]_{m_P}\}=M_P,\{[ij]'_1,\ldots ,[ij]'_{m_Q}\}=M_Q\subset\mathcal P^N_1$ need not be unique, in general,
but both $P$ and $Q$ display each a \textit{unique maximal} collection $\{[ij]_1,\ldots ,[ij]_{m^*_P}\}=M^*_P,\{[ij]'_1,\ldots ,[ij]'_{m^*_Q}\}=M^*_Q$
of atoms satisfying these two equalities. Clearly, $[ij]\not\leqslant\{A\}\cap P^{A^c}_{\bot}$
for all $[ij]\in M_P^*\cup M_Q^*$, and these two maximal collections have empty intersection, $M_P^*\cap M_Q^*=\emptyset$, in that
if there was any $[ij]$ included in both, then in partition $[ij]\vee P^A\cup P^{A^c}_{\bot}>^*P^A\cup P^{A^c}_{\bot}$ there would be
some $A'\supset A$ such that $P^{A'}=Q^{A'}$, and hence $A$ could \textit{not} be a maximal subset where $P$ and $Q$ coincide. Finally,
as $P\wedge Q=\underset{P\geqslant[ij]\leqslant Q}{\underset{[ij]\in\mathcal P^N_1}{\vee}}[ij]$, the sought conclusion
\begin{equation*}
\left\{[ij]\in\mathcal P^N_1:[ij]\leqslant P^A\cup P^{A^c}_{\bot}\right\}=\Big\{[ij]\in\mathcal P^N_1:[ij]\leqslant P\wedge Q\Big\}
\end{equation*}
follows.
\end{proof}
\smallskip

Thus, $s^P-s^{P^A}=s^P-s^{P\wedge Q}$ as well as $s^Q-s^{P^A}=s^Q-s^{P\wedge Q}$, and $\delta^{IH}(P,Q)=s^P+s^Q-2s^{P\wedge Q}=m^*_P+m^*_Q$, where $m^*_P,m^*_Q$ are as
above:
\begin{eqnarray*}
P=\left(P^A\cup P^{A^c}_{\bot}\right)\vee[ij]_1\vee\cdots\vee[ij]_{m^*_P}&\Rightarrow&s^P=s^{P\wedge Q}+m^*_P\text ,\\
Q=\left(P^A\cup P^{A^c}_{\bot}\right)\vee[ij]'_1\vee\cdots\vee[ij]'_{m^*_Q}&\Rightarrow&s^Q=s^{P\wedge Q}+m^*_Q\text .
\end{eqnarray*}
The issue is now constructing maximal collections $M^*_P,M^*_Q$ for maximizing or else minimizing $\delta^{IH}(P,Q)=|M^*_P|+|M^*_Q|$,
while obeying the following.
\begin{claim}
If $A\in 2^N,A\neq N$ is a maximal subset where $P^A=Q^A$, then
\begin{eqnarray}
\left(\underset{[ij]\in M^*_P\cup M^*_Q}{\bigcup}\{i,j\}\right)\cap A^c&=&A^c\text ,\\
\left(\underset{[ij]\in M^*_P\cup M^*_Q}{\bigcup}\{i,j\}\right)\cap A&\neq&\emptyset\text .
\end{eqnarray}
\end{claim}
\begin{proof}
The former condition seems evident: every $j\in A^c$ must be in the (unique) 2-cardinal block of at least one of the
atoms $[ij]\in M^*_P\cup M^*_Q$; otherwise, there would be some proper superset $A'\supset A$, namely the union of $A$ and
all $j\in A^c$ left out by both collections of atoms, were $P^{A'}=Q^{A'}$.\\
Now assume the latter condition is not satisfied: $\underset{[ij]\in M^*_P\cup M^*_Q}{\vee}[ij]\leqslant\{A^c\}\cup P^A_{\bot}$.
Define $P'=\underset{[ij]\in M^*_P}{\vee}[ij]$ and $Q'=\underset{[ij]\in M^*_Q}{\vee}[ij]$ and consider any $\hat A\subset A^c$
such that $|\hat A\cap B|=1$ for every $B\in(P'\vee Q')^{A^c}$. Then, $|\hat A|=|(P'\vee Q')^{A^c}|\geq 1$, and
$P^{A\cup\hat A}=Q^{A\cup\hat A}$, again violating the assumption that $A$ is a maximal subset where $P^A=Q^A$. 
\end{proof}

As every $j\in A^c$ has to be in the 2-cardinal block of at least one atom $[ij]\in M^*_P\cup M^*_Q$, minimization
surely attains, as long as possible, when every $j\in A^c$ is in the 2-cardinal block of \textit{precisely one} atom
in the union of the two maximal collections, in which case $\delta^{IH}(P,Q)=|M^*_P|+|M^*_Q|=D(P,Q)$.
\begin{claim}
For $k\leq 2(n-k)$, the lower bound is $\underset{D(P,Q)=k}{\underset{P,Q\in\mathcal P^N}{\min}}\delta^{IH}(P,Q)=k$.
\end{claim}
\begin{proof}
Fix $k\leq 2(n-k)$ and let $A\in 2^N,|A|=n-k$ be a maximal subset where the two partitions $P,Q$ to be constructed coincide.
Of course, if $k=0$, then $P=Q$ and there is nothing to show. Otherwise, for $k>0$, choose $P^A=P^A_{\bot}$. Then, inequality
$k\leq 2(n-k)$ entails that $P$ and $Q$ (satisfying (10), (11) and $M^*_P\cap M^*_Q=\emptyset$ above) may be constructed in a way such that they each admit
a \textit{unique} representation as a join of atoms (which is thus both the maximal and minimal one). As already observed, this
achieves when every $j\in A^c$ is in only one atom in the union of the two maximal collections, entailing, in turn, that: (1)
every $i\in A$ is in no more than two atoms in that union, and (2) such two (at most) atoms are one finer than $P$ and incomparable
with $Q$, and the other incomparable with $P$ and finer than $Q$. For the sake of concreteness, case $k=2(n-k)$ is easily detailed
by setting $A=\{i_1,\ldots ,i_{n-k}\}$ as well as $A^c=\{j_1,\ldots ,j_{2(n-k)}\}$. Then,
\begin{eqnarray*}
P&=&[i_1j_1]\vee[i_2j_3]\vee[i_3j_5]\vee\cdots\vee[i_mj_{2m-1}]\vee\cdots\vee[i_{n-k}j_{2(n-k)-1}]\text ,\\
Q&=&[i_1j_2]\vee[i_2j_4]\vee[i_3j_6]\vee\cdots\vee[i_mj_{2m}]\vee\cdots\vee[i_{n-k}j_{2(n-k)}]\text ,
\end{eqnarray*}
and $s^P=n-k=s^Q$ as well as $s^{P\wedge Q}=0$, hence $\delta^{IH}(P,Q)=2(n-k)=k$. In general, if the inequality is strict,
$k<2(n-k)$, then not all $n-k$ elements $i\in A$ appear in two atoms in the union of the two maximal collections, in that some
may be in only one atom, while some others may even be in no atom at all. What matters is that all needed conditions get satisfied
by making every $i\in A$ appearing in maximally two such atoms, each finer than one partition but incomparable with
the other, while if $j\in A^c$ appears precisely in one atom $[ij]\in M^*_P\cup M^*_Q$, then
$|A^c|=|M^*_P|+|M^*_Q|=\delta^{IH}(P,Q)=k=D(P,Q)$.
\end{proof}
\smallskip

For example, let $N=\{1,2,3,4,5,6\}$ and $P=[12]\vee[34]=\{12|34|5|6\}$ as well as
$Q=[56]=\{1|2|3|4|56\}$. Then, a maximal $A\subseteq N$ with $P^A=Q^A$ is 3-cardinal, say $A=\{2,4,6\}$, thus
$D(P,Q)=6-3=3$. Also, $s^{P\wedge Q}=0$ and $s^P=2=2s^Q$ yield $\delta^{IH}(P,Q)=2+1=3$. The same obtains for
$P=P_{\bot},Q=\{12|34|56\}$. Conversely, $P'=\{12|34|56\},Q'=\{16|23|45\}$ again yield $P'\wedge Q'=P_{\bot}$ and
a maximal $A\subseteq N$ with $P'^A=Q'^A$ such as $A=\{2,4,6\}$), but now $\delta^{IH}(P',Q')=6=2D(P',Q')$. On
the other hand, $P''=\{12|34|5|6\}$ and $Q''=\{12|35|4|6\}$ yield $P''\wedge Q''=[12]$ and there is a unique maximal
subset $A\subseteq N$ where $P''^A=Q''^A$; it is $A=\{1,2,4,5,6\}$, entailing $\delta^{IH}(P'',Q'')=2=2D(P'',Q'')$. 

If $k>2(n-k)$, then of course the above construction does not yield the same result, but still indicates how to obtain the sought
minimum: basically, either $P$ or $Q$ or both constructed in that manner display some block with cardinality $\geq 3$. In particular,
the construction remains valid for determining two \textit{minimal} collections of atoms whose join yields the two partitions
$P,Q$ where IH distance is minimized. In the union of these two minimal collections, every $j\in A^c$ (with $A\in 2^N$ being
a maximal subset where $P^A=Q^A$) still compares in precisely one atom, but when turning to maximal collections this is no longer achievable.  
\begin{claim}
For $2(n-k)<k<n-1$, the lower bound is $\underset{D(P,Q)=k}{\underset{P,Q\in\mathcal P^N}{\min}}\delta^{IH}(P,Q)=$
\begin{equation*}
=\left(k-(n-k)\left\lfloor\frac{k}{n-k}\right\rfloor\right)\left(\binom{\left\lfloor\frac{\left\lceil\frac{k}{n-k}\right\rceil}{2}\right\rfloor+1}{2}+
\binom{\left\lceil\frac{\left\lceil\frac{k}{n-k}\right\rceil}{2}\right\rceil+1}{2}\right)+
\end{equation*}
\begin{equation*}
+\left(n-2k+(n-k)\left\lfloor\frac{k}{n-k}\right\rfloor\right)\left(\binom{\left\lfloor\frac{\left\lfloor\frac{k}{n-k}\right\rfloor}{2}\right\rfloor+1}{2}+
\binom{\left\lceil\frac{\left\lfloor\frac{k}{n-k}\right\rfloor}{2}\right\rceil+1}{2}\right)\text .
\end{equation*}
\end{claim}
\begin{proof}
Again choose $P^A=P^A_{\bot}$ consisting of $|A|=n-k$ singletons or 1-cardinal blocks. Then, covering $A^c$ with atoms $[ij]$ or pairs $\{i,j\}$ such
that $i\in A,j\in A^c$ as indicated above entails that some elements $i\in A$ have to be atom-linked with more than two distinct elements $j,j'\in A^c$,
while every $j\in A^c$ still appears in precisely one atom. Making this as uniform as possible, every $i\in A$
appears in either $\left\lfloor\frac{k}{n-k}\right\rfloor$ or else $\left\lceil\frac{k}{n-k}\right\rceil$ atoms in the union of the two collections.
Then, the best every $i\in A$ can do for minimizing distance $\delta^{IH}(P,Q)$ while being atom-linked with either $\left\lfloor\frac{k}{n-k}\right\rfloor$
or else $\left\lceil\frac{k}{n-k}\right\rceil$ distinct $j\in A^c$, is splitting these $\left\lfloor\frac{k}{n-k}\right\rfloor$ or else
$\left\lceil\frac{k}{n-k}\right\rceil$ atoms equally, or as equally as possible, between $P$ and $Q$. That is,
$\left\lfloor\left\lfloor\frac{k}{n-k}\right\rfloor/2\right\rfloor$ or else
$\left\lfloor\left\lceil\frac{k}{n-k}\right\rceil/2\right\rfloor$ finer than $P$ but incomparable with $Q$, while the remaining
$\left\lceil\left\lfloor\frac{k}{n-k}\right\rfloor/2\right\rceil$ or else $\left\lceil\left\lceil\frac{k}{n-k}\right\rceil/2\right\rceil$ ones
incomparable with $P$ but finer than $Q$. Each of these four cardinalities $m\in\left\{\left\lfloor\left\lfloor\frac{k}{n-k}\right\rfloor/2\right\rfloor,
\left\lfloor\left\lceil\frac{k}{n-k}\right\rceil/2\right\rfloor,
\left\lceil\left\lfloor\frac{k}{n-k}\right\rfloor/2\right\rceil,
\left\lceil\left\lceil\frac{k}{n-k}\right\rceil/2\right\rceil\right\}$ corresponds to the formation of a block, either in $P$ or in $Q$, whose cardinality
is $m+1$, precisely because $m$ is a number of elements $j\in A^c$ to which a common $i\in A$ is joined through atoms
$[ij]$. Finally, the number of elements $i\in A$ appearing in $\left\lceil\frac{k}{n-k}\right\rceil$ atoms is
$k-(n-k)\left\lfloor\frac{k}{n-k}\right\rfloor$, while $(n-k)-\left[k-(n-k)\left\lfloor\frac{k}{n-k}\right\rfloor\right]$ is the number of
elements $i\in A$ appearing in $\left\lfloor\frac{k}{n-k}\right\rfloor$ atoms.
\end{proof}
\smallskip

For example, let $N=\{1,2,3,4,5,6,7\}$ and fix $A=\{1,2\}$ as a maximal subset such that $P^A=Q^A=\{1|2\}$. Then, $2(n-k)=2(7-5)=4<5=k$
and $1\in A$ has to be atom-linked with $\lceil\frac{7-2}{2}\rceil=3$ elements $j\in A^c$ while $2\in A$ has to be atom-linked with
$\lfloor\frac{7-2}{2}\rfloor=2$ elements $j\in A^c=\{3,4,5,6,7\}$ (or vice versa switching 1 and 2). On the other hand,
$1\in A$ divides these three atoms into two determining (through join) partition $P$ and the remaining one determining partition $Q$.
Similarly, $2\in A$ being involved in a even number of atoms, these latter can be divided equally between $P$ and $Q$. This means
\begin{eqnarray*}
P&=&[13]\vee[14]\vee[26]=[13]\vee[14]\vee[34]\vee[26]=\{134|26|5|7\}\text ,\\
Q&=&[15]\vee[27]=\{15|27|3|4|6\}\text .
\end{eqnarray*}
Hence $\delta^{IH}(P,Q)=s^P+s^Q-2s^{P\wedge Q}=4+2-0=6>5=7-2=D(P,Q)$. Conversely, if $A=\{1,2,3\}$ is a maximal subset with
$P^A=Q^A=\{1|2|3\}$, then $2(n-k)=2(7-4)=6>4=k$ and thus the situation is that of claim 11. Accordingly, partitions $P,Q$ may be
(for example) as follows:
\begin{eqnarray*}
P&=&[14]\vee[26]=\{14|26|3|5|7\}\text ,\\
Q&=&[15]\vee[37]=\{15|2|37|4|6\}\text ,
\end{eqnarray*}
yielding $\delta^{IH}(P,Q)=4=D(P,Q)$.

Finally, case $k=n-1$ is simple: conditions (10), (11) and $M^*_P\cap M^*_Q=\emptyset$ entail that in one of the two partitions,
say $P$, all $n-1$ elements $j\in A^c$
are atom-linked with the unique element $\{i\}=A$, entailing $P=P^{\top}$, while the other partition has to be $Q=P_{\bot}$. On the other hand,
$\delta^{IH}(P,Q)=\binom{n}{2}$ if and only if $P=P^{\top},Q=P_{\bot}$. Therefore,
\begin{equation*}
D(P,Q)=n-1\Leftrightarrow P=P^{\top},Q=P_{\bot}\Leftrightarrow\delta^{IH}(P,Q)=\binom{n}{2}\text .
\end{equation*}

For the upper bound all the above conditions (21), (22) and $M^*_P\cap M^*_Q=\emptyset$ remain valid, but $M^*_P\cup M^*_Q$
must be as large as possible. To achieve this, rather than distributing the needed atoms $[ij]\in M^*_P\cup M^*_Q, i\in A,j\in A^c$
in the most uniform way over the $n-k$ elements $i\in A$ as for the lower bound, it is now necessary to concentrate them as much as possible,
which is easy. 
\begin{claim}
$\underset{D(P,Q)=k}{\underset{P,Q\in\mathcal P^N}{\max}}\delta^{IH}(P,Q)=\binom{n}{2}-\binom{n-k}{2}$ for all $0\leq k\leq n-1$.
\end{claim}
\begin{proof}
Let $A\in 2^N,|A|=n-k$ be a maximal subset such that $P^A=Q^A$. Again, if $k=0$, then $P=Q$ and there is nothing to show.
Otherwise, for $0<k\leq n-1$, choose $P=P^{\top},Q=\{A\}\cup P^{A^c}_{\bot}$, entailing $P^A=\{A\}$. Then,
$\delta^{IH}(P,Q)=s^P+s^Q-2s^Q=s^P-s^Q=\binom{n}{2}-\binom{n-k}{2}$. It seems rather evident that there is no way
of satisfying (10), (11) and mostly $M^*_P\cap M^*_Q=\emptyset$ while involving a larger number of atoms. 
\end{proof}
\smallskip

Note that this upper bound attains on pairs of comparable modular elements, and on such pairs the IH and SB distances coincide. In fact,
at each level $\mathcal P^N_k,0\leq k<n$ of the partition lattice the size function attains its maximum precisely on modular elements:
$\underset{P\in\mathcal P^N_k}{\max}$ $s^P=\binom{k+1}{2}$. 

\subsection{Constrained bounds}
In these lower and upper bounds considered above the cardinality $n-k$ of a maximal subset $A\in 2^N$ where the two generic partitions
$P,Q$ coincide is fixed, while the form of $P^A=Q^A$ is chosen arbitrarily. In fact, for the lower bound the choice is $P^A=P^A_{\bot}$
and for the upper one it is $P^A=\{A\}$. Accordingly, the constrained version of these bounding problems also fixes $P^A$, through its
class $c^{P^A}=\left(c^{P^A}_1,\ldots ,c^{P^A}_{n-k}\right)$. Considering such a version may be useful for further seeing in detail how
many distinct values are actually taken by $\delta^{IH}(\cdot,\cdot)$ for each value $k,0\leq k<n$ of partition-distance
$D(\cdot,\cdot)$.

By claim 12 above, determining the constrained upper bound is simple:
\begin{equation*}
\underset{\underset{\underset{P^A\text{ FIXED}}{D(P,Q)=k}}{P,Q\in\mathcal P^N}}{\max}
\delta^{IH}(P,Q)=\underset{\underset{B\neq B'}{B,B'\in P^A}}{\max}\left(\binom{|B|+k}{2}-\binom{|B|}{2}+\binom{|B'|+k}{2}-\binom{|B'|}{2}\right)
\end{equation*}
with $k=|A^c|$. That is, partition $P$ chooses a largest block $B\in P^A$ and obtains as
$P=\left(P^A\cup P^{A^c}_{\bot}\right)\underset{\underset{i\in B,j\in A^c}{[ij]\in\mathcal P^N_1}}{\vee}[ij]$, while $Q$
chooses a largest block $B'\in P^A\backslash B$ and obtains as
$Q=\left(P^A\cup P^{A^c}_{\bot}\right)\underset{\underset{i\in B',j\in A^c}{[ij]\in\mathcal P^N_1}}{\vee}[ij]$. Accordingly, the distance is the sum of
$s^P-s^{P\wedge Q}=\binom{|B|+k}{2}-\binom{|B|}{2}$ and $s^Q-s^{P\wedge Q}=\binom{|B'|+k}{2}-\binom{|B'|}{2}$.

Like in the free version of the problem, determining the constrained lower bound is less simple, in general. Still, an immediate adaptation of 
claim 10 to this more general situation is: if $k\leq 2c^{P^A}_1$, then
$\underset{\underset{\underset{P^A\text{ FIXED}}{D(P,Q)=k=|A^c|}}{P,Q\in\mathcal P^N}}{\min}\delta^{IH}(P,Q)=k$.\\
The lower bound for \textit{all} other cases where $k>2c^{P^A}_1$ clearly cannot be approached by considering separatley all possible classes $c^{P^A}$
of partitions of a $n-k$-set, with $n$ arbitrarily large and $0<k<n$. Conversely, what seems interesting is an algorithmic view of the
problem. In particular, the sought lower bound may be determined through a greedy construction of a bipartite graph $G=(V\times V',E)$ where
$V=P^A,V'=A^c,E\subseteq V\times V'$. In words, vertex subset $V$ contains all blocks of $P^A$, vertex subset $V'$ contains all elements
$j\in A^c$ and any edge $(B,j)\in E$ links a block $B\in P^A$ and an element $j\in A^c$. To see how this relates to the constrained bounding problem,
firstly let the graph $G_{\emptyset}=(V\times V',\emptyset)$ with empty edge set correspond to the initial situation where
$P_0=P^A\cup P^{A^c}_{\bot}=Q_0$, with $P>P_0,Q>Q_0$ denoting the two partitions to be constructed by adding edges and such that, eventually,
$\delta^{IH}(P,Q)$ is the sought lower bound. Now consider adding edges one after the other, while conceiving edge set $E=E^P\overset{\cdot}{\cup}E^Q$
as partitioned into two blocks $E^P,E^Q$ corresponding to partitions $P,Q$. This yields a sequence
$G_0=G_{\emptyset},G_1=(V\times V',E_1^P\overset{\cdot}{\cup}E_1^Q),\ldots ,G_m=(V\times V',E_m^P\overset{\cdot}{\cup}E_m^Q),\ldots$
of bipartite graphs. In particular, let the sequence of edges progressively added one after the other be such that all links added at an odd step
$m=1,3,5,\ldots$ are in $E^P$, while all links added at an even step $m=2,4,6,\ldots$ are in $E^Q$. Finally, the main rule for the construction is the following:
at the end, every vertex $j\in A^c$ has to be the end-vertex of \textit{precisely one} added edge. Evidently, this means that the above sequence terminates exactly
at the $k$-th step, where $k=|A^c|$, and allows for some blocks $B\in P^A$ to remain isolated vertexes in the final graph.

At any intermediate step $m,0<m<k$, graph $G_m=(V\times V',E_m^P\overset{\cdot}{\cup}E_m^Q)$
identifies the two (not yet final) partitions $P_m,Q_m$ as follows:
\begin{equation*}
P_m=P_0\underset{\underset{i\in B}{(B,j)\in E^P_m}}{\vee}[ij]\text{ as well as }
Q_m=Q_0\underset{\underset{i'\in B'}{(B',j)\in E^Q_m}}{\vee}[i'j]\text .
\end{equation*}
Define $e_{G_m}^P,e_{G_m}^Q:V\rightarrow\mathbb Z_+$ by
\begin{equation*}
e_{G_m}^P(B)=\big|\big\{(B,j):(B,j)\in E^P_m\big\}\big|\text{ and }
e_{G_m}^Q(B)=\big|\big\{(B,j):(B,j)\in E^Q_m\big\}\big|
\end{equation*}
for every $B\in V$. Then, $\delta^{IH}(P_m,Q_m)=s^{P_m}+s^{Q_m}-2s^{P_0}=$
\begin{equation*}
\sum_{B\in V}\left(\binom{|B|+e_{G_m}^P(B)}{2}+\binom{|B|+e_{G_m}^Q(B)}{2}-2\binom{|B|}{2}\right)\text .
\end{equation*}

Now, for any sequence $G_0,G_1,\ldots ,G_m,0\leq m<k$ of bipartite graphs as above, define a weight function
$w_{m+1}(\cdot)$ (over edges) as follows: if $m=0$ or $m$ is even, then $w_{m+1}:(V\times V')\backslash E_m^P\rightarrow\mathbb N$ assigns an integer weight to
every edge $(B,j)\not\in E_m^P$ by
\begin{equation*}
w_{m+1}((B,j))=\binom{|B|+1+e_{G_m}^P(B)}{2}\text ,
\end{equation*}
while if $m$ is odd, then $w_{m+1}:(V\times V')\backslash E_m^Q\rightarrow\mathbb N$ assigns an integer weight to
every edge $(B,j)\not\in E_m^Q$ by
\begin{equation*}
w_{m+1}((B,j))=\binom{|B|+1+e_{G_m}^Q(B)}{2}\text .
\end{equation*}
This enables to construct a $k+1$-sequence of graphs $G_0,G_1^*,\ldots G^*_m,\ldots ,G^*_k$ in a greedy fashion, that is, by adding at each step
$m,0<m\leq k$ an edge with minimum weight as given by weight function $w_m(\cdot)$. Then, the sought lower bound is
\begin{equation*}
\sum_{B\in V}\left(\binom{|B|+e_{G^*_k}^P(B)}{2}+\binom{|B|+e_{G^*_k}^Q(B)}{2}-2\binom{|B|}{2}\right)\text ,
\end{equation*}
and every vertex $B\in V$ which remains isolated
in the final graph $G^*_k$ has $e_{G^*_k}^P(B)=0=e_{G^*_k}^Q(B)$, entailing that the corresponding term in the summation simply vanishes.


\section{Concluding remarks}
Quantifying differences between partitions is needed in statistics, where partitions are clusterings and blocks are clusters. On the other hand,
the partition lattice is very important in lattice theory, where it appears to be the main example of an indecomposable geometric lattice.
While the Hamming distance between subsets may be extended to any distributive lattice, how to measure differences between elements of
geometric lattices seems disregarded in combinatorial theory. This paper addresses the issue from alternative perspectives, and in general shows that
any monotone lattice function, such as the rank, may be used for constructing a distance measure.

Among the measures considered, the IH distance clearly is the analog of the Hamming distance between subsets. It obtains by focusing on atoms and through the
size function. In particular, the size and the rank of a partition are, respectively, the maximum and the minimum number of atoms whose join yields that
partition; their difference maximally is $s^{P^{\top}}-r(P^{\top})=\binom{n-1}{2}$. Conversely, for every subset there is a unique number of atoms whose
join (union) yields that subset, and this number is the rank of the subset.

Given that a variety of distance measures is considered, comparing them seems natural as well as useful, and to this end any distance measure
$\delta(\cdot,\cdot)$ may be $[0,1]$-normalized as $\hat\delta(\cdot,\cdot)=1-[1/(1+\delta(\cdot,\cdot))]$. In this view, the larger the range
(or image) $\underset{(P,Q)\in\mathcal P^N\times\mathcal P^N}{\cup}\hat\delta(P,Q)=R(\hat\delta)\subset[0,1]$ of a normalized
distance, the more precise and granular this latter is. Pushing the comparison into a ranking, the less attractive (normalized) distances appearing above
are those satisfying modularity and co-maximality, hence the rank-based $\hat\delta^{RB}$ distance and the size-based $\hat\delta^{SB}$ one. Apart from
their range, these two measures are not able to appreciate that non-modular partitions have many complements, some of which (strictly) coarser than
others \cite{Stanley71}, and this is a main flaw. Next come the modified RB distance $\hat\delta^{RB}_+$ and partition-distance $\hat D$; their range is
rather small if the aim is at distinguishing between all possible differences between partitions. The SD distance $\hat\delta^{SD}$ has a larger range, but still
smaller than IH distance $\hat\delta^{IH}$.

From a final perspective, determining $D(P,Q)$ for generic partitions $P,Q$ is a computational problem whose solution requires polynomial time
\cite[theorem 2.1, p. 160]{Gusfield02}. On the other hand, if prepared to use binary $\binom{n}{2}$-arrays as data structures, then IH distance is
$\delta^{IH}(P,Q)=\langle I_P,I_Q\rangle$, where $\langle\cdot,\cdot\rangle$ denotes scalar product while $I_P:\mathcal P^N_1\rightarrow\{0,1\}$ is
the indicator function (or binary $\binom{n}{2}$-array representation)
of partitions $P\in\mathcal P^N$ introduced above.

\bibliographystyle{abbrv}
\bibliography{biblioPartitionDistances}

\end{document}